\documentclass[aps,
	       pra,
	       reprint,
	       twocolumn,
	       superscriptaddress,
	       ]{revtex4-1} 

\usepackage[english]{babel}
\usepackage[T1]{fontenc} 
\usepackage{times}

\usepackage{graphics}
\usepackage{paralist} 
\usepackage{enumitem} 
\setlist{nolistsep} 
\usepackage{array} 
\usepackage{amsthm} 
\usepackage{amssymb} 
\usepackage{mathtools} 
\usepackage{amsmath} 
\usepackage{mathbbol} 
\usepackage{complexity} 

\usepackage{color}
\usepackage{hyperref} 
\hypersetup{colorlinks=true}

\newtheorem{theorem}{Theorem}
\newtheorem{lemma}[theorem]{Lemma}
\newtheorem{corollary}[theorem]{Corollary}
\newtheorem{observation}[theorem]{Observation}
\newtheorem{definition}[theorem]{Definition}
\newtheorem{problem}[theorem]{Problem}

\newcommand{\specialcell}[2][c]{
  \begin{tabular}[#1]{@{}c@{}}#2\end{tabular}}
  
\newcolumntype{I}{!{\vrule width 1pt}} 
\newlength\savedwidth
\newcommand\whline{\noalign{\global\savedwidth\arrayrulewidth\global\arrayrulewidth 1pt}%
\hline
\noalign{\global\arrayrulewidth\savedwidth}}

\DeclareMathOperator{\Tr}{Tr} 

\DeclareMathOperator{\diag}{diag}
\DeclareMathOperator{\rank}{rank}

\renewcommand{\Pr}{\operatorname{Pr}}

\newcommand{\1}{\mathbb{1}}
\newcommand{\NN}{\mathbb{N}}
\newcommand{\RR}{\mathbb{R}}
\newcommand{\CC}{\mathbb{C}}
\newcommand{\ZZ}{\mathbb{Z}}
\newcommand{\QQ}{\mathbb{Q}}
\newcommand{\KK}{\mathbb{Q}}
\newcommand{\CCalg}{\QQ}
\newcommand{\RRalg}{\QQ}

\newcommand{\mc}[1]{\mathcal{#1}}

\newcommand{\ket}[1]{\left|{#1}\right\rangle}
\newcommand{\ketn}[1]{| #1 \rangle}
\newcommand{\ketb}[1]{\bigl| #1 \bigr\rangle}
\newcommand{\bra}[1]{\left\langle{#1}\right|}

\newcommand{\ketbra}[2]{\ket{#1} \!\! \bra{#2}}

\newcommand{\sandwich}[3]{\left\langle  #1 \right| #2 \left| #3 \right\rangle}
\newcommand{\sandwichb}[3]{\bigl\langle  #1 \bigr| #2 \bigl| #3 \bigr\rangle}

\newcommand{\ad}{^\dagger}

\newcommand{\restr}{\upharpoonright}

\newcommand{\argdot}{ \,\cdot \, }

\newcommand{\ls}{S}
\newlang{\MPOTP}{TP}
\newlang{\BMPOTP}{BTP}
\newlang{\bMPOTP}{(B)TP}
\newlang{\bPCP}{(B)PCP}
\newlang{\BPCP}{BPCP}
\newclass{\distNP}{distNP}

\newcommand{\titletext}{Matrix product operators and states: NP-hardness and undecidability}

\newcommand{\listpacs}{03.67.-a, 02.60.Pn, 03.65.Ud, 89.70.Eg}

\newcommand{\fu}{QMIO Group, Dahlem Center for Complex Quantum Systems, Freie Universit{\"a}t Berlin, 14195 Berlin, Germany}
\newcommand{\uf}{Physikalisches Institut, Universit\"{a}t Freiburg, 79104 Freiburg, Germany}

\newcommand{\listkeys}{%
MPO, 
matrix product density operator, MPDO,
realization problem,
finitely correlated state, FCS,
local purification,
threshold problem,
bounded Post correspondence problem,
positive matrix factorization,
positive semidefinite rank,
undecidable tensor problem
}

\makeatletter 
 \hypersetup{pdftitle={\titletext},
	     pdfauthor={M. Kliesch, D. Gross, J. Eisert},
	     pdfsubject={PACS numbers: \listpacs},
	     pdfkeywords={\listkeys}
	    }
\makeatother
	     
\begin{document}

\title{\titletext}

\author{M.\ Kliesch} \affiliation{\fu}
\author{D.\ Gross} \affiliation{\uf}
\author{J.\ Eisert} \affiliation{\fu}



\begin{abstract}
Tensor network states constitute an important variational set of quantum states for numerical studies of strongly correlated systems in condensed-matter physics, as well as in mathematical physics. This is specifically true for finitely correlated states or matrix-product operators, designed to capture mixed states of one-dimensional quantum systems. It is a well-known open problem to find an efficient algorithm that decides whether a given matrix-product operator actually represents a physical state that in particular has no negative eigenvalues. We address and answer this question by showing that the problem is provably undecidable in the thermodynamic limit and that the bounded version of the problem is \NP-hard in the system size. Furthermore, we discuss numerous connections between tensor network methods and (seemingly) different concepts treated before in the literature, such as hidden Markov models and tensor trains. 
\end{abstract}

\maketitle

\section{Introduction}
Computational quantum many-body physics is marred by the fact that standard computational descriptions of states require exponentially many parameters. 
Fortunately, for many physically relevant problems, one does not need to consider all those parameters to capture natural properties extremely accurately. 
One of the pillars on which computational many-body approaches rest is the framework of tensor network methods. 
Here, the relevant degrees of freedom are parameterized by very few numbers, which are organized in terms of tensor networks that are contracted in order to compute expectation values
\cite{Schollwoeck,Rommer,Review,Reviews1,Reviews2,Reviews3,WolfReview,AreaTensorNetworks}. 
Notably, the \emph{density matrix renormalization group approach}, the most successful method to numerically determine ground state properties of strongly correlated one-dimensional models, can be cast into such a form \cite{Schollwoeck,Rommer}. 
In this language, the problem of minimizing the energy can be phrased as a variational principle over \emph{matrix product} (or \emph{purely generated $C^\ast$-finitely correlated}) \emph{states} \cite{FanNacWer92}. 
The natural analogue that also encompasses mixed quantum states are \emph{matrix product operators}. 
Again, they feature strongly in numerical algorithms \cite{VerGarCir04,ZwoVid04}, for example when investigating stationary states of local Liouvillians modelling open quantum systems \cite{Ban13,CaiBar13} or Gibbs states \cite{Gibbs,Barthel}. 

However, general matrix product operators are not guaranteed to represent physical states, which is the source of considerable conceptual and computational difficulties. 
It would thus be highly desirable to design an efficient algorithm capable of checking whether a given matrix product representation defines a positive operator. 
To decide if such an efficient ``local test for positivity'' exists is a fundamental problem in the field, implicit already in its early formulations \cite[Appendix]{FanNacWer92}.

Here, we address and answer this question: 
Determining whether a matrix product operator defines a physical state in the thermodynamic limit is a provably \emph{undecidable problem}. 
We also show that the bounded version of the problem is \NP-hard in the number of tensors, burying hopes that one could find an efficient algorithm testing for positivity exactly. 
This is proven for quantum spin chains with local dimension $d=2$ by a polynomial reduction from the \emph{Post correspondence problem} and a bounded variant thereof. 

To give a practical example, one can approximate stationary states of local Liouvillians by iteratively applying the Liouvillian to a state described as a matrix product operator and subsequently truncating the tensors. 
To avoid inconsistent results, one has to check whether the truncation step has caused the state to become ``too unphysical'' in that it has created eigenvalues that are more negative than some chosen tolerance threshold. 
We prove such a check to be unfeasible. 
The practical implications of our work are as follows: One hand they motivate the quest for finding specific feasible instances that might exist.
This quest is reminiscent of the task of finding, e.g., efficient contractions of two-dimensional planar tensor networks, even though this task has been identified to be $\#\P$-complete \cite{SchuchPEPSContraction}.
On the other hand, it shows that one should direct one's efforts towards finding approximate solutions.

The insight presented here adds a natural many-body problem to the list of quantum mechanical questions that have recently been identified not only as computationally hard, but as outright undecidable \cite{EisMueGog12,WolCubPer11,CubPerWol14,MorBia12}.
Along the way of introducing these novel results, we discuss a number of connections between concepts that have arisen in the literature, but whose relation has received surprisingly little attention (see Table~\ref{tab:concepts}). 

\begin{table}
\begin{tabular}{I c I c | c I}
\whline States
& Classical & Quantum 
\\ \whline
\specialcell{Pure} & 
\specialcell{Deterministic\\ finite automata\\ \cite{HopMotUll07}} & 
\specialcell{\phantom{$|^|$}Matrix-product states \cite{Review}, 
\\ purely generated $C^\ast$-finitely \\correlated states \cite{FanNacWer92},
\\ tensor trains \cite{TT}}
\\ \hline
\specialcell{Mixed,\\inherently \\positive} & 
\specialcell{Hidden Markov \\ models \cite{Vid11},\\ probabilistic finite\\ automaton \cite{BloCat03}} & 
\specialcell{$C^\ast$-finitely correlated states \cite{FanNacWer92},\\ local purification \cite{CueSchPer13}, 
\\ quantum Markov chains \cite{Guta}}
\\ \hline
\specialcell{Mixed,\\ not inherently\\ positive} & 
\specialcell{Quasi-realizations\\ \cite{Vid11}} & 
\specialcell{Finitely correlated states \cite{FanNacWer92}, \\ matrix product density\\ operators \cite{VerGarCir04,ZwoVid04}}
\\ \whline
\end{tabular}
\caption{Concepts of tensor networks discussed here.}\label{tab:notions}\label{tab:concepts}
\end{table}

\section{Tensor networks} 
In quantum many-body theory, tensor network methods are widely used in order to avoid
intractability problems. 
The idea is to resort to variational classes of states, where the attention is restricted to low-dimensional manifolds of states that seem to capture well the relevant physics of the model under study. 
It is less widely appreciated in the physics community that similar structures are ubiquitous in classical probability theory: 
The hidden Markov model (HMM) is a generalization of a Markov chain, where the observable process $(Y_t)_{t \in \NN}$ does not need to be Markovian but there is a stochastic process $(X_t)_{t\in \NN}$ carrying additional information that renders the combined process $(X_t, Y_t)$ Markovian. 
We only consider the case where $X_t$ and $Y_t$ have finitely many outcomes and call the number of outcomes of $X_t$ the \emph{bond dimension} $D$. 
With transition probabilities 
$M^{(\alpha)}_{i,j} = \Pr[(X_{t+1},Y_{t+1}) = (j,\alpha) \mid X_{t} = i]$, 
boundary condition $p_j = \Pr[X_1 = j]$, and $\underline{1} \coloneqq (1,1,\dots,1)^T\in \RR^D$, the probabilities of outcome sequences of the process $Y_t$ are given by the matrix product
\begin{equation}\label{eqn:hmm}
	\Pr[Y_1 = \alpha_1, \dots, Y_n = \alpha_n]
	= p \, M^{(\alpha_1)} \dots M^{(\alpha_n)} \underline 1 \, .
\end{equation}
In order for an HMM to describe a stationary process, $p$ is usually taken to be a stationary distribution, i.e., $ \sum_{\alpha=1}^d M^{(\alpha)} p = p$. 
The description complexity of the HMM is independent of $n$, or, if we allow the Markov kernels $M$ to vary as a function of $t$, linear. 
Non-negativity of the probabilities in Eq.~\eqref{eqn:hmm} is guaranteed because they arise as the contraction over element-wise non-negative tensors. 

From Eq.~\eqref{eqn:hmm}, it follows that the matrices $F^{(k,n)}$ defined by 
\begin{equation}\nonumber 
	F^{(k,n)}_{(\alpha_1, \dots, \alpha_k),(\alpha_{k+1}, \dots, \alpha_{k+n})}
	\coloneqq
	\Pr[Y_1 = \alpha_1, \dots, Y_{k+n} = \alpha_{k+n}]
\end{equation}
have rank at most $D$, which upper bounds the so-called \emph{Hankel-rank} \cite{Vid11}. 

A natural question is whether the rank condition alone characterizes those distributions that allow for a HMM with bond dimension $D$. 
It has been known since the 1960s that this is \emph{not} the case: There are distributions where $\rank (F^{(k,n)} )\leq D$ for all $k,n$, yet no HMM with finite bond dimension exists \cite{FoxRub68,DhaNad70}. 
However, a relatively straight-forward argument (based on sequential ``rank-revealing decompositions'', e.g., singular value decompositions) shows that every distribution with rank bounded by $D$ allows for a representation as in Eq.~\eqref{eqn:hmm} where the tensors $M$ are \emph{not necessarily positive}. 
This fact seems to have been discovered independently in different contexts, e.g., Refs.\ \cite{HMMSVD,FanNacWer92,VidalSVD,WolfReview,TT}. 
The resulting form is known as a \emph{quasi-realization}, offering the same concise description of the distribution as a HMM. 
These are, however, more difficult to work with computationally, as any variation of the local tensors can destroy global positivity. 
An important question thus is as follows: Are the conditions on the tensor $M$ that guarantee global positivity computationally efficiently verifiable? 
As we will see, the answer is \emph{no}. 

\begin{figure}
\includegraphics{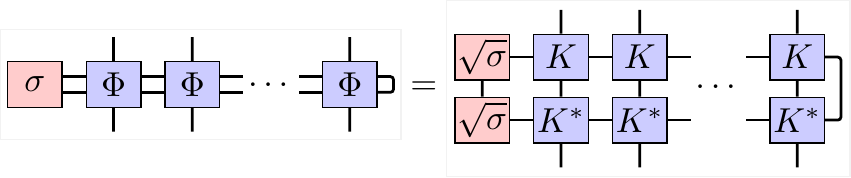} 
\caption{
A $C^\ast$-FCS as a tensor network. 
The channel $\Phi$ can be written in terms of Kraus operators, $\Phi(\rho) = \sum_i K_i\ad \rho K_i$. 
The vertically contracted indices between $K$ and $K^\ast$  correspond to the sum over $i$. 
The tensor network to the right is referred to as a local purification.
}\label{fig:C-star-FCS}
\end{figure}

The above constructions generalize to the quantum setting: 
A \emph{$C^*$-finitely correlated state}  \cite{FanNacWer92} (a.k.a.\ \emph{quantum Markov chain} \cite{Guta,Gohm}) $\rho$ is obtained by replacing the elements of Eq.~\eqref{eqn:hmm} by their quantum counterparts: 
We substitute $p$ by a $D\times D$ density matrix $\sigma$, the stochastic map $M$ by a completely positive map $\Phi$ that maps states on $\CC^D$ to those on $\CC^D \otimes \CC^d$, and $\underline{1}$ by the partial trace (cp.\ Figure~\ref{fig:C-star-FCS}). 
This immediately yields a \emph{local purification} \cite{CueSchPer13}: 
One can write $\Phi$ in Kraus representation by choosing operators 
$K_i:\CC^D \to \CC^D \otimes \CC^d$ 
satisfying $\sum_{i=1}^E K_i K_i\ad = \1$ and 
$\Phi(\argdot) = \sum_{i=1}^{E} K_i\ad \argdot K_i$. 
Here, $E$ can be assumed to be smaller than or equal to $d\, D$. 
Then the $n$-fold application of $K$ to a purification $\ket{\sqrt \sigma}$ of $\sigma$ yields the state vector $\ket\psi = K\circ K \circ \dots \circ K \ket{\sqrt \sigma}$ in  
$\CC^D \otimes \left(\CC^d \otimes \CC^{E}\right)^{\otimes n} \otimes \CC^D$ that is a local purification of $\rho$, see Figure~\ref{fig:C-star-FCS}. 
For quantum states one can, once more, define ``quasi-representations''. 
Here $\Phi$ can be a general linear map and $\sigma$ is some operator (no positivity constraints). 
This results in what is known as a \emph{matrix product density operator} (MPDO) or \emph{finitely correlated state} (FCS) (not $C^*$-FCS). 
A discussion of different notions of positivity is provided in the Appendix. 
More concretely, with $[d] \coloneqq \{1,2, \dots, d\}$, 
an MPDO is a density matrix written in the following form: 


\begin{definition}[Matrix product operator (MPO)] 
An instance of \emph{MPO-tensors} is given by 
$M = \bigl(M^{(\alpha,\beta)}_{i,j}\bigr)_{\alpha,\beta \in [d], \, i,j \in [D]} \in \CC^{d\times d\times D\times D}$ 
and \emph{boundary vectors} $\ket L, \ket R \in \CC^{D}$. 
The dimension $d$ is called the \emph{physical dimension} and $D$ the \emph{(MPO)-bond dimension.}
The generated translation invariant \emph{MPO} for \emph{system size} $n$ is
\begin{equation}\nonumber
	\rho(L,M,R,n)\coloneqq \sum_{j \in [D]^{n+1}}
	L_{j_1} M_{j_1,j_2}\otimes \dots \otimes M_{j_{n},j_{n+1}} R_{j_{n+1}}.
\end{equation}
\end{definition}

\section{Main results}
In order to precisely define the problems that are shown to be computationally unfeasible, we employ the standard language of theoretical computer science: The task of identifying objects with a certain property (e.g., those MPOs that are positive) is a \emph{decision problem}. 
A specific case (e.g., given by a concrete tensor and boundary vectors) is an \emph{instance}. 
A decision problem is \emph{\NP-hard} if it is at least as hard (in a precise sense) as all other problems from the complexity class \NP. 
It is deemed highly unlikely that any \NP-hard problem can be solved efficiently on either a classical or a quantum computer. 
A problem is \emph{(Turing) undecidable} if no computer, even if endowed with unbounded resources, is capable of correctly solving all instances. 
In the statements of the various problems below, MPO-tensors are specified by rational numbers. 
These have finite descriptions and can thus serve as inputs to computer programs. 
Allowing for more general numbers (e.g.\ complex rationals) would make the problem only harder. 

In the precise statement of the problem, we allow for a threshold $\lambda$ which bounds the ``degree of negativity'' that is deemed acceptable. 
Moreover, we call positive semi-definite operators more concisely just \emph{positive}. 

\begin{problem}[Bounded MPO threshold problem (\BMPOTP)]\label{p:BMPOTP}
\label{p:positivity} \hfill 
\begin{description}[noitemsep,leftmargin=\widthof{Question:},font=\normalfont]
 \item [Instance] MPO-tensors $M \in \CCalg^{d\times d\times D\times D}$, $\ket L, \ket R \in \CCalg^D$, threshold $\lambda\in\RRalg$, and system size $n$.
 \item [Question] Is the MPO $\rho(L,M,R,n)+\lambda \1$ positive?
\end{description}
\end{problem}

\begin{problem}[MPO threshold problem (\MPOTP)]\label{p:MPOTP}
The \MPOTP{} is defined in the same way as the \BMPOTP{} except that there is no restriction on the system size and the question is: 
Is there an $n \in \ZZ^+$ such that $\rho(L,M,R,n)+\lambda \1$ is \emph{not} positive? 
\end{problem}

We obtain the following results, where the latter one adapts ideas from Ref.~\cite{BloCat03} to the quantum setting:

\begin{theorem}[\NP-hardness of the bounded MPO threshold problem]\label{thm:hardness}
For any $\lambda \in \QQ$ and physical dimension $d\geq2$, the \BMPOTP{} is \NP-hard.
\end{theorem}

\begin{theorem}[Undecidability of the MPO threshold problem]\label{thm:undecidability}
For each threshold $\lambda\in \RRalg$ the \MPOTP{} is undecidable. 
In particular, this holds for the case where the physical dimension is $d=2$, the bond dimension is $D=42$, and the matrices $M_{i,j}$ are diagonal for all $i,j =1,\dots, D$. 
\end{theorem}

\section{Outlook}
An important question is whether there are physically relevant instances for which positivity is efficiently decidable and how this can be exploited best in numerical algorithms. 
Sometimes one can, e.g., efficiently detect negativity locally by calculating expectation values with respect to matrix product states (MPS) of small bond dimension, see Figure~\ref{fig:MPS_test}. 
\begin{figure}
\includegraphics{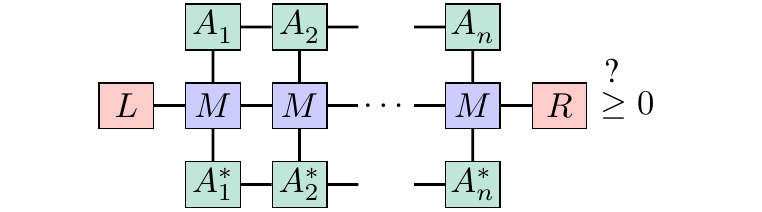}
\caption{Contracting MPOs with MPS can detect negativity for some instances. Hence, this provides a hierarchy of efficient tests labeled by the MPS-bond dimension, a strategy 
practically accessible by DMRG approaches.
}\label{fig:MPS_test}
\end{figure}

In Ref.~\cite{CueSchPer13} local purifications of positive MPOs in terms of matrix product states are investigated and it is shown that the arising MPS-bond dimension can in general not be bounded independently of the system size. 
This already suggest that such purifying MPS would require high bond dimensions when used instead of MPOs in numerical simulations.
However, two constructive purification methods are suggested that are efficient when the rank of the MPO is polynomially bounded but in general necessarily inefficient \cite{CueSchPer13}. 
From our Theorem~\ref{thm:hardness} it also follows that this is no coincidence. 
To be more precise, a \emph{local purification method} is an algorithm that receives an MPO instance $M, \ket{L}, \ket{R}$ and a system size $n$ with $\rho(L,M,R,n)\geq 0$ as input and outputs a local purification of $\rho(L,M,R,n)$. 

\begin{corollary}
 Local purification methods are inefficient in the system size (assuming that $\P\neq \NP$). 
\end{corollary}

In the \BMPOTP{} one is asked to \emph{exactly} delineate the MPOs with smallest eigenvalues above $-\lambda$ from those with smallest eigenvalues below $-\lambda$. 
In practice, it would be acceptable if an algorithm reliably recognizes whether a state $\rho$ is either sufficiently positive, i.e., $\rho \geq -\lambda$, or violates a threshold by at least $\epsilon\geq 0$, i.e., $\rho \ngeq -(\lambda +\epsilon)$. 
Such an approximate version is allowed to give unspecified results on the narrow band between the two cases. 
In order to make this precise, we state the \BMPOTP{} as a \emph{weak membership problem}: 
For $\epsilon>0$ one is only required to decide whether an MPO instance $(L,M,R)$ with $\Tr(\rho(L,M,R,n))=1$ satisfies either
$\rho(L,M,R,n) \geq - \lambda$ or 
$\rho(L,M,R,n) \ngeq -(\lambda+\epsilon)$.
The MPO provided in the proof of Theorem~\ref{thm:hardness} has a trace that is exponentially bounded from above. Hence, as a corollary, one obtains that the \BMPOTP{} remains \NP-hard as a weak membership problem if $\epsilon$ is exponentially small in $n$. 
This statement remains true for algebraic and not necessarily rational inputs. 
Weak membership formulations seem to be natural for a variety of problems in quantum information. 
For instance, \NP-hardness of testing separability of quantum states as a weak membership problem was established first \cite{Gur03} for an exponentially small ``error'' $\epsilon$ and, much later \cite{Gha10}, for a polynomially
small $\epsilon$, in fact, using key methods of the previous approach \cite{Gur03}. 
Hence, our work is an invitation to explore whether the \BMPOTP{} as a weak membership problem is also \NP-hard for only polynomially bounded $\epsilon$ or, instead, to actually find an algorithm that efficiently solves that problem. 

\section{Details}
For any finite set $\Sigma$ (alphabet) we denote by $\Sigma^n$ the set of sequences $w$ (words) of $n=|w|$ elements (letters) from $\Sigma$ and by $\Sigma^\ast \coloneqq \bigcup_{n \in \NN} \Sigma^n$ the set of words. 
For $w \in [d]^\ast$ we denote by 
$\ket w = \bigotimes_{j=1}^{|w|} \ket{w_j}$ the tensor product of the corresponding canonical basis states $\ket 1, \dots, \ket d$. 
Given two monoids $W$ and $W'$ we call a map $U:W \to W'$ a \emph{morphism} if $U$ maps the identity element of $W$ to the identity element of $W'$ and $U(w_1 w_2) = U(w_1) U(w_2)$ for all $w_1,w_2 \in W$. 
The monoids we encounter here are either given by words over an alphabet with concatenation or by matrices with matrix multiplication. 
Next, we introduce the famous \emph{Post correspondence problem} \cite{PCP} and a bounded variant thereof. 

\begin{problem}[Bounded Post correspondence problem (\BPCP)]\label{p:BPCP}
\hfill \vspace{-\baselineskip}
 \begin{description}[noitemsep,leftmargin=\widthof{Question:},font=\normalfont]
  \item [Instance] Pairs of words $(u_\alpha,v_\alpha)_{\alpha \in [d]}$ over a finite alphabet $\Sigma$ and length $n$ in unary notation \cite{Footnote}. 
  \item [Question] Does there exist a non-empty word $w \in [d]^n$ of length $n$ such that $u_{w_1} u_{w_2} \dots u_{w_{n}} = v_{w_1} v_{w_2} \dots v_{w_{n}}$? 
 \end{description}
\end{problem}

\begin{problem}[Post correspondence problem (\PCP)] \label{p:PCP}
 The \PCP{} is defined in the same way as the \BPCP{} except that there is no restriction on the word length. 
\end{problem}

The two sets of words $(u_\alpha,v_\alpha)_{\alpha \in [d]}$, referred to as \emph{dominos}, define two morphisms
$U, V : [d]^\ast \to \Sigma^\ast$ given by $U(w) = u_{w_1}u_{w_2} \dots u_{w_{|w|}}$ and similar for $V$. 

\begin{theorem}[\PCP{} is undecidable \cite{MatSen05}]
\label{thm:PCP_undecidable}
 For every $d\geq 7$ the $\PCP$ with $\Sigma=\{0,1\}$ is undecidable. 
\end{theorem}

In fact, by noting a simpler proof \cite{Sip05} of a variant of this theorem with larger $d$, one can make the following computation theoretic statement: 
  
\begin{observation}[\BPCP{} is \NP-complete]\label{obs:PCP_NP}
There is a polynomial $p$ such that for any non-deterministic Turing machine $\mc{M}$ and input $x$ there is a reduction to an instance $U,V$ of the \BPCP{} such that $\mc{M}$ accepts $x$ in $n$ steps iff there is a solution of $U,V$ of length $p(n)$. 
\end{observation}

In the usual textbook proof of the undecidability of the \PCP{} (see, e.g., Ref.~\cite{Sip05}) the halting problem is reduced to the \PCP. 
The idea of the proof is to encode the computation history of a given Turing machine into the two morphisms of the \PCP{} in two different ways such that there is a solution iff the Turing machine halts. 
Specifically, there is a polynomial $p$ such that exactly when a Turing machine accepts an input after $n$ steps then there is a solution of the corresponding \PCP{} instance of length $p(n)$. 
The encoding works in the same way when the Turing machine is replaced by a non-deterministic Turing machine because having a transition relation instead of a transition function allows us to define the \PCP{} instances in the same way (possibly with more ``dominos'' $(u_\alpha,v_\alpha)$). 
This shows that if one could solve the \BPCP{} in polynomial time, then one could also solve \NP-problems in polynomial time. 
As solutions to the \BPCP{} can be verified by a Turing machine in polynomial time the \BPCP{} is in \NP. 

For the proofs of Theorems~\ref{thm:hardness} and~\ref{thm:undecidability} the two following polynomial reductions are needed (see the Appendix for proofs building on Refs.~\cite{BloCat03,Hir06,BloTsi97}):

\begin{lemma}\label{lem:d6}
Let $(u_j,v_j)_{j \in [d]}$ be any instance of the \PCP{} and 
 $\lambda \in \RRalg$ be a threshold. 
 Then there exist boundary vectors $\ket L, \ket R$, and matrices 
 $A^{(1)}, \dots, A^{(d)} \in \NN^{6 \times 6}$ that define a morphism 
 $A:[d]^\ast \to \NN^{6 \times 6}$ such that for all $w \in [d]^\ast$
 \begin{equation}\label{eq:Asandwich}
 \begin{aligned}
 \sandwich{L}{A^{(w)}}{R} &=  -(\lambda+1)  &&\text{if}&& U(w)=V(w),\\
 \sandwich{L}{A^{(w)}}{R} &\geq -\lambda    &&\text{if}&& U(w)\neq V(w) .
 \end{aligned}
 \end{equation}
\end{lemma}

\begin{lemma}[\cite{BloTsi97}]
\label{lem:kto2}
 Let $d,D \geq 2$, 
 $A^{(1)}, \dots, A^{(d)}\in \KK^{D\times D}$ be matrices that define a morphism 
 $A: [d]^\ast \to \KK^{D\times D}$, and $\ket L, \ket R \in \KK^{D}$ be boundary vectors. 
 Then there exist two matrices $B^{(1)}$ and $B^{(2)}$ that define a morphism 
 $B: [2]^\ast \to \KK^{D d\times D d}$ together with an injective morphism 
 $X: [d]^\ast \to [2]^\ast$ satisfying 
 $|X(w)| = d |w|$ such that 
 $\sandwichb{\tilde L}{B^{(X(w))}}{\tilde R} = \sandwich{L}{A^{(w)}}{R}$
 for all $w \in [d]^\ast$, 
 where $\ketn{\tilde X} \coloneqq (\ket X, 0,\dots, 0)^T$.
\end{lemma}

\begin{proof}[Proof of Theorems~\ref{thm:hardness} and~\ref{thm:undecidability}]
We prove the theorem by using the encoding $A :[d]^\ast \to \RRalg^{D\times D}$ of the \PCP{} with $d$ ``dominos'' into the matrices from Lemma~\ref{lem:d6}. 
Using Lemma~\ref{lem:kto2} we reduce the physical dimension $d$ to $2$ at the expense of having a larger bond dimension $d D$ and an increase of the system size $n$ to $d\, n$. 
This results in an encoding $C:[2]^\ast \to \RRalg^{dD\times dD}$ with boundary vectors $\ketb{\tilde L}$ and $\ketb{\tilde R}$. 
Now we define an MPO-tensor $M$ by 
$M^{(\alpha,\beta)} = \sum_{\gamma = 1}^d \delta_{\alpha,\gamma} \delta_{\beta,\gamma} C^{\gamma} $. 
Then $(\ket L, M, \ket R)$ is an encoding of \PCP{} to \MPOTP. 
All successively used encodings are polynomial reductions. 
In particular, an instance of \BPCP{} with word length $n$ can be written as an instance of \BMPOTP{} with system size $d n$. 
Hence, Theorem~\ref{thm:PCP_undecidable} and Observation~\ref{obs:PCP_NP} finish the proof. 
\end{proof}

\section{Conclusions}
In this work, we have shown that a problem naturally occurring in the context of tensor network states is \NP-hard and in the thermodynamic limit even undecidable. 
The findings point to the challenge for reliable numerical methods for, e.g., finding Gibbs and stationary states of quantum many-body systems: 
Truncations in the bond dimension --- a common step in existing numerical algorithms --- can introduce inconsistencies that cannot be found computationally. 
This insight provides an interesting twist to numerical methods to capture mixed quantum many-body systems as well as to notions of Hamiltonian complexity \cite{HamiltonianComplexity1,HamiltonianComplexity2}. 
Future research should follow a dual aim: First, identify instances and approximations where (near-)pos\-itivity can be guaranteed; second, search for further problems in the context of tensor network states that are not decidable algorithmically. 

\section{Acknowledgments}
The work of M.K.\ and J.E.\ is supported by the Studienstiftung des Deutschen Volkes, the FQXi, the EU
(RAQUEL, SIQS), and the ERC (TAQ). 
D.G.\ acknowledges support from the Excellence Initiative of the German Federal and State Governments (Grant ZUK 43), from the US Army Research Office under Contracts W911NF-14-1-0098 and W911NF-14-1-0133 and the Freiburg Research Innovation Fund. 

\newpage
\section{Appendix}
\appendix
In the first part of this appendix we discuss different kinds of postive matrix factorizations, positive ranks, and cones. 
In the last part we prove Lemmas~\ref{lem:d6} and~\ref{lem:kto2} using ideas from Refs.~\cite{BloCat03,BloTsi97}. 

\section{Matrix factorizations}
\label{sec:matrixfactorisations}
In this section we first comment on the gap between the usual \emph{rank} and the \emph{non-negative rank} of a matrix. 
Then we briefly discuss another notion of positive rank, called \emph{semidefinite rank}. 

\paragraph*{Non-negative rank.} 
The problem of finding HMM realizations is closely connected with the theory of \emph{non-negative matrix factorizations}. 
Any entry-wise non-negative matrix $F$, e.g., given by $F_{\alpha,\beta} = \Pr[Y_1 = \alpha, Y_2 = \beta]$, can be written as 
\begin{equation}
 F = \sum_{j=1}^{D} \ketbra{L_i}{R_i}, 
\end{equation}
i.e., as a sum of rank-$1$ matrices. This is a quasi-realization of the process $(Y_1, Y_2)$ with two time steps. 
The minimal $D$ is the usual rank of $F$. 
Similarly, the \emph{positive rank} of $F$ is the minimum number of entry-wise non-negative rank-$1$ matrices that sum up to $F$. 
Any such decomposition into rank-$1$ matrices with the proper normalization is a HMM for $(Y_1, Y_2)$. 

We thus find that the auxiliary dimension $D$ of a quasi-realization of two discrete random variables is nothing but the rank of the joint probability distribution seen as a matrix; 
while the smallest $D$ in a HMM is the non-negative rank. 
Old results \cite{FoxRub68,DhaNad70} showing that minimal bond dimensions for realizations and quasi-realizations are distinct can thus be re-interpreted in terms of gaps between rank and non-negative rank. 
Actually finding the non-negative rank is known to be \NP-hard \cite{vavasiv_complexity_2009}. 

Non-negative matrix factorizations have been studied extensively, just to name two famous examples, in Ref.~\cite{lee_learning_1999} for influential applications in machine learning theory, and in Ref.~\cite{lee_algorithms_2000} for algorithms.

\paragraph*{Positive semidefinite rank.} 
If $A_\alpha, B_\beta$ are non-negative matrices, then $\Tr(A_\alpha B_\beta)$ is clearly a non-negative number. 
One defines a \emph{positive semidefinite (PSD) factorization} of a matrix $F$ to be a choice of positive semidefinite matrices $A_\alpha\geq 0, B_\beta \geq 0$ such that 
\begin{equation}
	F_{\alpha,\beta} = \Tr( A_\alpha B_\beta).
\end{equation}
The \emph{PSD rank} of $F$ is the smallest dimension $D$ such that there is a PSD factorization of $F$ in terms of $D\times D$-matrices. 
These notions have recently attracted considerable attention
\cite{yannakakis_expressing_1991,fiorini_exponential_2011,gouveia_lifts_2011} 
as novel lower-bounds on the PSD rank (partly derived in the study of quantum communication complexity \cite{wolf_nondeterministic_2003}) have been used to disprove the existence of efficient semidefinite
programs that would solve certain \NP-hard combinatorial optimization problems \cite{fiorini_exponential_2011}. The problem of identifying semidefinite factorizations for empirically observed statistics of quantum experiments has been treated in Ref.~\cite{Cyril13} and prior works by the same author referenced there.

\paragraph*{Bond dimensions of MPOs and local purifications.}
\label{sec:local_purifications}
As in the previous case, one can easily check that a PSD factorization for $F$ 
immediately gives a realization of the \emph{classical} bivariate distribution defined by $F$ in terms of a \emph{quantum} HMM, i.e., a local purification.
Of course, classical distributions can be embedded into quantum states.  
This was used in Ref.~\cite{CueSchPer13} to leverage the fact \cite{fiorini_exponential_2011} that there are families of $F$s whose rank is bounded, but whose PSD rank diverges to show that there are bi-partite quantum states with bounded bond dimensions as MPOs, but unbounded ones for any local purification.

\section{Cones and generalizations}
\label{sec:cones}
One can abstract even further. Element-wise positivity and semidefiniteness are two notions of positivity in vector spaces. 
There is a systematic theory of general ordered vector spaces \cite{cones_and_duality}. 
Notions of order in real vector spaces stands in one-one correspondence to (Archimedean) closed convex cones. 
The cone represents those elements of the vector space that are deemed ``non-negative''. For two vectors $v,w\in V$, one says $v\leq w$  iff $v-w$ is an element of that cone. 
One can now define generalized HMMs by the ``physical'' and the ``auxiliary'' vector spaces carrying various orders and by requiring that the $F$s are positive maps in the sense that they preserve these orders. 
We can then distinguish various cases. 
If both physical and classical spaces carry the usual element-wise order, we recover the traditional notion of HMM. If both carry the semidefinite order, we obtain quantum Markov processes. 
If the physical one is element-wise, but the auxiliary one semidefinite, one arrives at the model treated in Ref.~\cite{CueSchPer13} and also in this work. 

Such \emph{conal matrix factorizations} have recently been discussed in the optimization literature 
\cite{fiorini_exponential_2011,gouveia_lifts_2011,FioMasPat13}. 
Notably, researchers working on stochastic process both in a classical \cite{Dha63} and quantum \cite[Appendix]{FanNacWer92} setting have appreciated these ideas already quite early. 
They are also inherent to the \emph{generalized probabilistic theories} approach to quantum foundations \cite{Bar07,BarBarLei06,Hol82,GroMueCol10}.

\section{Proof of Lemmas~\ref{lem:d6} and~\ref{lem:kto2}}
\label{sec:proofs}
\begin{proof}[Proof of Lemma~\ref{lem:d6}]
For an alphabet $\Sigma$ of size $b=|\Sigma|$ the \emph{numeric representation} of $\Sigma^\ast$ is the map 
$\sigma: \Sigma^\ast \to \NN$ given by
\begin{equation}
 \sigma(w) \coloneqq \sum_{j = 1}^{|w|} \sigma(w_j) b^{|w|-j} 
\end{equation}
where $\sigma(\emptyset) = 0$ and $\sigma\!\!\restr_\Sigma : \Sigma \to [b]$ is a bijection (enumeration of $\Sigma$). 
Note that for any two words $u,v \in \Sigma^\ast$
\begin{equation}\label{eq:sigma_rule}
 \sigma(uv) = b^{|v|}\sigma(u) + \sigma(v)
\end{equation}
and $\sigma(u) = \sigma(v)$ iff $u=v$. 
Note that for all $w$
\begin{equation}
 b^{{|w|}-1}\leq \sigma(w) \leq b^{|w|} -1 .
\end{equation}

Next, we define $A: \Sigma^\ast \times \Sigma^\ast \to \NN^{6\times 6}$ as in Refs.~\cite{BloCat03,Hir06} by
\begin{equation}\nonumber
 A(u,v) = \begin{pmatrix}
             b^{2 |u|} & 0 & 0 & 0 & 0 & 0\\
             0 & b^{|u|+|v|} & 0 & 0 & 0 & 0\\
             0 &  0 & b^{2|v|} & 0 & 0 & 0\\
             \sigma(u) b^{|u|} & \sigma(v) b^{|u|} & 0 & b^{|u|} & 0 & 0 \\
             0 & \sigma(u) b^{|v|} & \sigma(v) b^{|v|} & 0 & b^{|v|} & 0\\
             \sigma(u)^2 & 2 \sigma(u)\,\sigma(v) & \sigma(v)^2 & 2 \sigma(u) & 2 \sigma(v) & 1 
          \end{pmatrix}.
\end{equation}
Using the property \eqref{eq:sigma_rule} of the numeric representation of $\Sigma^\ast$ it follows that 
\begin{equation}\label{eq:A_morphism_property}
 A(u_1,v_1)\, A(u_2,v_2) = A(u_1 u_2, v_1 v_2)
\end{equation}
for all words $u_1, u_2, v_2, v_2 \in \Sigma^\ast$. 
Now let $U,V: [d] \to \Sigma$ be an instance of the \PCP{} and define 
$A^{(\alpha)} \coloneqq A(U(\alpha),V(\alpha))$ for all $\alpha \in [d]^\ast$ 
and the boundary vectors  as $\ket L = \ket 6$ and 
$\ket R = \ket 1 - \ket 2 + \ket 3 -(\lambda+1) \ket 6$. 
Thanks to Eq.~\eqref{eq:A_morphism_property}, $A$ viewed as a map on $[d]$ extends to a morphism on $[d]^\ast$, as required. 
For any $w \in [d]^\ast$ we have $U(w) = V(w)$ iff $A^{(w)}_{6,4} = A^{(w)}_{6,5}$. 
Moreover,
\begin{equation}
 \sandwich{L}{A^{(w)}}{R} 
 = \bigl(\sigma(U(w)) - \sigma(V(w))\bigr)^2 -(\lambda+1) \, ,
\end{equation}
and Eq.~\eqref{eq:Asandwich} follows.
\end{proof}

\begin{proof}[Proof of Lemma~\ref{lem:kto2}]
We using a technique from Ref.~\cite{BloTsi97}. 
Define $B^{(1)} \coloneqq \diag\left(A^{(\alpha)} \right)_{\alpha \in [d]}$ 
to be the block diagonal matrix having $A^{(\alpha)}$ in increasing order as the diagonal blocks and
\begin{equation}
 B^{(2)} \coloneqq \begin{pmatrix}
                    0 & \1_{D(d-1)} \\
                    \1_D & 0 
                   \end{pmatrix}. 
\end{equation}
Now we construct a bijective morphism $X: [d]^\ast \to [2]^\ast$ satisfying 
\begin{equation}
 \sandwichb{\tilde L}{B^{(X(w))}}{\tilde R}
 =
 \sandwich{L}{A^{(w)}}{R} .
\end{equation}
First, note that $B^{(2)} = \ls \otimes \1_D$, where $\ls$ is the permutation matrix representing the cyclic permutation $(1,2,\dots, d) \mapsto (2,3,\dots, d,1)$. 
This means that $B^{(2)}$ acts as the cyclic left shift on the $d$ blocks of $D$ neighbouring components of column vectors. 
Hence 
\begin{align}
 C^{(\alpha)} &\coloneqq {B^{(2)}}^{\alpha-1} \, B^{(1)}\, {B^{(2)}}^{d-(\alpha-1)}
 \\
 &=\diag(A^{(\alpha)}, \dots, A^{(d)}, A^{(1)}, \dots, A^{(\alpha-1)}) \label{eq:Cw}
\end{align}
for all $\alpha \in [d]$. 
Eq.~\eqref{eq:Cw} implies that for any $w \in [d]^\ast$ the upper left $D\times D$ block of $C^{(w)}$ is $A^{(w)}$. 
Next, define $\ketb{\tilde L}, \ketb{\tilde R} \in \ZZ^{dD}$ to be the vectors that have $\ket L$ and $\ket R$ as their first $D$ components respectively and all other components equal to zero. Then we obtain
\begin{equation}\label{eq:CA}
 \sandwichb{\tilde L}{C^{(w)}}{\tilde R} = \sandwich{L}{A^{(w)}}{R}
\end{equation}
for all $w \in [d]^\ast$. 

For $s \in \NN$ and $\alpha \in [d]$ we denote by $\alpha^s$ the word that consists of $s$ times the letter $\alpha$. 
Now we define a morphism $X: [d]^\ast \to [2]^\ast$ by
\begin{equation}
 X(\alpha) \coloneqq 2^{\alpha-1} 1\, 2^{d-\alpha}
\end{equation}
for $\alpha \in [d]$. Note that for any $w \in [d]^\ast$ we have $|X(w)| = d |w|$ and that one can reconstruct  $w$ from just knowing $X(w)$, i.e., $X$ is injective. 
Using Eq.~\eqref{eq:Cw} we obtain
\begin{equation}
 C^{(w)} = B^{(X(w))}.
\end{equation}
Together with Eq.~\eqref{eq:CA} this finishes the proof. 
\end{proof}



\begin{thebibliography}{99}

\bibitem{Rommer} 
	S.\ Rommer and S.\ Ostlund,
	\href{http://dx.doi.org/10.1103/PhysRevB.55.2164}
	  {Phys.\ Rev.\ B {\bf 55}, 2164 (1997).}

\bibitem{Schollwoeck} 
	U.\ Schollwoeck,
	\href{http://dx.doi.org/10.1016/j.aop.2010.09.012}
	  {Ann.\ Phys.\ {\bf 326}, 96 (2011).}

\bibitem{Review} 
	F.\ Verstraete, J.~I.\ Cirac, and V.\ Murg,
	\href{http://dx.doi.org/10.1080/14789940801912366}
	  {Adv.\ Phys.\ {\bf 57}, 143 (2008).}
	
\bibitem{Reviews1} 
	N.\ Schuch, \href{http://arXiv.org/abs/1306.5551}{arXiv:1306.5551.}
	
\bibitem{Reviews3} 
	R.\ Orus, 
	\href{http://dx.doi.org/10.1016/j.aop.2014.06.013}
	  {Ann.\ Phys.\ (N.Y.) {\bf 349}, 117 (2014).}
	
\bibitem{Reviews2} 	
	J.\ Eisert, 
	Modeling and Simulation {\bf 3}, 520 (2013);
	\href{http://arxiv.org/abs/1308.3318}{arXiv:1308.3318.}
	
\bibitem{WolfReview}	
	D.\ Perez-Garcia, F.\ Verstraete, M.~M.\ Wolf, and J.~I.\ Cirac,
	Quant.\ Inf.\ Comput.\ {\bf 7}, 401 (2007); 
	\href{http://arxiv.org/abs/quant-ph/0608197}{arXiv:quant-ph/0608197.}
	
\bibitem{AreaTensorNetworks} 
	J.\ Eisert, M.\ Cramer, and M.~B.\ Plenio, 
	\href{http://dx.doi.org/10.1103/RevModPhys.82.277}
	  {Rev.\ Mod.\ Phys.\ {\bf 82}, 277 (2010);}
	\href{http://arxiv.org/abs/0808.3773}
	  {arXiv:0808.3773.}
	
\bibitem{FanNacWer92} 
	M.\ Fannes, B.\ Nachtergaele, and R.~F.\ Werner,
	\href{http://dx.doi.org/10.1007/BF02099178}
	  {Comm.\ Math.\ Phys.\ {\bf 144}, 429 (1992).}

\bibitem{VerGarCir04} 
	F.\ Verstraete, J.~J.\ Garcia-Ripoll, and J.~I.\ Cirac, 
	\href{http://dx.doi.org/10.1103/PhysRevLett.93.207204}
	  {Phys.\ Rev.\ Lett.\ {\bf 93}, 207204 (2004);} 
	\href{http://arxiv.org/abs/cond-mat/0406426}
	  {arXiv:cond-mat/0406426.}
	
\bibitem{ZwoVid04} 
	M.\ Zwolak and G.\ Vidal,
	\href{http://dx.doi.org/10.1103/PhysRevLett.93.207205}
	  {Phys.\ Rev.\ Lett.\ {\bf 93}, 207205 (2004);}
	\href{http://arxiv.org/abs/cond-mat/0406440}{arXiv:cond-mat/0406440.}

\bibitem{Ban13}
	M.~C.\ Banuls, {\it Tensor network techniques for the study of dissipative dynamics},
	talk held in Stellenbosch, April 19 (2013).

\bibitem{CaiBar13} 
	Z.\ Cai and T.\ Barthel, 
	\href{http://dx.doi.org/10.1103/PhysRevLett.111.150403}
	  {Phys.\ Rev.\ Lett.\ {\bf 111}, 150403 (2013);}
	\href{http://arxiv.org/abs/1304.6890}{arXiv:1304.6890.}

\bibitem{Barthel} 
	T.\ Barthel, U.\ Schollw{\"o}ck, and S.~R. White,
	\href{http://dx.doi.org/10.1103/PhysRevB.79.245101}
	  {Phys.\ Rev.\ B {\bf 79}, 245101 (2009).}

\bibitem{Gibbs}
	E.~M.\ Stoudenmire and S.~R.\ White,
	\href{http://dx.doi.org/10.1088/1367-2630/12/5/055026}
	  {New J.\ Phys.\ {\bf 12}, 055026 (2010).}

\bibitem{SchuchPEPSContraction} 
	N.\ Schuch, M.~M.\ Wolf, F.\ Verstraete, J.~I.\ Cirac,
	\href{http://dx.doi.org/10.1103/PhysRevLett.98.140506}
	  {Phys.\ Rev.\ Lett.\ {\bf 98}, 140506 (2007).}
	
\bibitem{EisMueGog12} 
	J.\ Eisert, M.~P.\ Muller, and C.\ Gogolin,
	\href{http://dx.doi.org/10.1103/PhysRevLett.108.260501}
	  {Phys.\ Rev.\ Lett.\ {\bf 108}, 260501 (2012).}

\bibitem{WolCubPer11} 
	M.~M.\ Wolf, T.~S.\ Cubitt, and D.\ Perez-Garcia,
	\href{http://arXiv.org/abs/1111.5425}{arXiv:1111.5425.}
	
\bibitem{MorBia12} 
	J.\ Morton and J.\ Biamonte,
	\href{http://dx.doi.org/10.1103/PhysRevA.86.030301}
	  {Phys.\ Rev.\ A  {\bf 86}, 030301(R) (2012).}

\bibitem{CubPerWol14} 
	T.~S.\ Cubitt, D.\ Perez-Garcia, and M.~M. Wolf,
	{\it Undecidability of the spectral gap problem},
	talk given at XVII Conference on Quantum Information Processing, Barcelona, Feb.\ 6 (2014).

\bibitem{Vid11} 
	M.\ Vidyasagar, 
	\href{http://dx.doi.org/10.1007/s00498-011-0066-7}
	  {Math.\ Control Signals Syst.\ {\bf 23}, 1 (2011).}

\bibitem{FoxRub68} 
	M.\ Fox and H.\ Rubin, 
	\href{http://dx.doi.org/10.1214/aoms/1177698326}
	  {Ann.\ Math.\ Stat.\ {\bf 39}, 938 (1968).}

\bibitem{DhaNad70} 
	S.~W.\ Dharmadhikari and M.~G.\ Nadkarni, 
	\href{http://dx.doi.org/10.1214/aoms/1177697202}
	  {Ann.\ Math.\ Stat.\ {\bf 41}, 207 (1970).}
	
\bibitem{HMMSVD}
	R.~V.\ Erickson, 
	\href{http://dx.doi.org/10.1214/aoms/1177696962}
	  {Ann.\ Math.\ Stat.\ {\bf 41} 843 (1970).}

\bibitem{VidalSVD} 
	G.\ Vidal, 
	\href{http://dx.doi.org/10.1103/PhysRevLett.93.040502}
	  {Phys.\ Rev.\ Lett.\ {\bf 93}, 040502 (2004).}
	
\bibitem{TT} 
	I.~V.\ Oseledets, 
	\href{http://dx.doi.org/10.1137/090752286}
	  {SIAM J.\ Sci.\ Comput.\ {\bf 33}, 2295 (2011);}
	I.~V.\  Oseledets and E.~E.\ Tyrtyshnikov, 
	\href{http://dx.doi.org/10.1137/090748330}
	  {SIAM J.\ Sci.\ Comput.\ {\bf 31}, 3744  (2009).}
		
\bibitem{Guta} 
	M.\ Guta and J.\ Kiukas, 
	\href{http://arXiv.org/abs/1402.3535}{arXiv:1402.3535.}
	
\bibitem{Gohm}
	O.~E.\ Barndorff-Nielsen, U.\ Franz, R.\ Gohm, B.\ K{\"u}mmerer, and S.\ Thorbj{\o}rnsen,
	{\it Quantum independent increment processes II} (Springer, Heidelberg, 2006).

\bibitem{CueSchPer13} 
	G.\ De las Cuevas, N.\ Schuch, D.\ P\'erez-Garc\'ia, and J.~I.\ Cirac,
	\href{http://dx.doi.org/10.1088/1367-2630/15/12/123021}
	  {New J.\ Phys.\ {\bf 15}, 123021 (2013).}

\bibitem{BloCat03} 
	V.~D.\ Blondel and V.\ Catarini, 
	\href{http://dx.doi.org/10.1007/s00224-003-1061-2}
	  {Theory Comput.\ Syst.\ {\bf 36}, 231 (2003).}

\bibitem{Gur03} 
	L.\ Gurvits, 
	\href{http://dx.doi.org/10.1145/780542.780545}
	  {Proc.\ 35th ACM Symp.\ Th.\ Comp.\ {\bf 10} (2003).}
	
\bibitem{Gha10}
	S.\ Gharibian, 
	Quant.\ Inf.\ Comp.\ {\bf 10}, 343 (2010);
	\href{http://arxiv.org/abs/0810.4507}{arXiv:0810.4507.}

\bibitem{PCP} 
	E.\ Post, 
	\href{http://dx.doi.org/10.1090/S0002-9904-1946-08555-9}
	  {Bull.\ Amer.\ Math.\ Soc.\ {\bf 52}, 264  (1946).}

\bibitem{Footnote}
	``Unary notation'' effectively means that the problem size of the input $n$ equals $n$.
  
\bibitem{MatSen05} 
	Y.\ Matiyasevicha and G.\ S\'enizergues, 
	\href{http://dx.doi.org/10.1016/j.tcs.2004.09.016}
	  {Theoret.\ Comput.\ Sci.\ {\bf 330}, 145 (2005).}

\bibitem{Sip05} 
	M.\ Sipser, 
	``{\it Introduction to the Theory of Computation}'', 
	\href{http://www.icsd.aegean.gr/kaporisa/index_files/sipser.pdf}
	{2nd ed. (Thomson Course Technology, Boston, 2005).}

\bibitem{Hir06} 
	M.\ Hirvensalo, 
	\href{http://dx.doi.org/10.1007/978-3-540-69507-3_25}
	  {LNCS {\bf 4362}, 309 (2007).}

\bibitem{BloTsi97} 
	V.~D.\ Blondel and J.~N.\ Tsitsiklis,
	\href{http://dx.doi.org/10.1016/S0020-0190(97)00123-3}
	  {Inform.\ Process.\ Lett.\ {\bf 63}, 283 (1997). }

\bibitem{HopMotUll07} 
	J.~E.\ Hopcroft, R.\ Motwani, J.~D.\ Ullman, 
	{\it Introduction to automata theory, languages, and computation} (Pearson Addison Wesley, 2007).

\bibitem{HamiltonianComplexity1} 
	T.~J.\ Osborne, 
	\href{http://arXiv.org/abs/1106.5875}{arXiv:1106.5875.}

\bibitem{HamiltonianComplexity2}	
	D.\ Aharonov, I.\ Arad, Z.\ Landau, and U.\ Vazirani, 
	\href{http://arXiv.org/abs/1011.3445}{arXiv:1011.3445.}
	
\bibitem{vavasiv_complexity_2009} 
	S.~A.\ Vavasis, 
	\href{http://dx.doi.org/10.1137/070709967}
	  {SIAM J.\ Optim.\ {\bf 20}, 1364 (2009).}

\bibitem{lee_learning_1999}
	D.~D.\ Lee and H.~S.\ Seung, 
	\href{http://dx.doi.org/10.1038/44565}
	  {Nature {\bf 401}, 788 (1999).}

\bibitem{lee_algorithms_2000}
	D.~D.\ Lee and H.~S.\ Seung, 
	\href{http://papers.nips.cc/paper/1861-algorithms-for-non-negative-matrix-factorization}
	  {Adv.\ Neural Inform.\ Process.\ Systems {\bf 13}, 556 (2001).}

\bibitem{fiorini_exponential_2011} 
	S.\ Fiorini, S.\ Massar, S.\ Pokutta, H.~R.\ Tiwary, and R.\ de Wolf, 
	\href{http://dx.doi.org/10.1145/2213977.2213988}
	  {Proc.\ ACM STOC {\bf 44}, 95 (2012);}
	\href{http://arxiv.org/abs/1111.0837}{arXiv:1111.0837.}

\bibitem{yannakakis_expressing_1991}
	M.\ Yannakakis, 
	\href{http://dx.doi.org/10.1145/62212.62232}
	  {Proc.\ ACM STOC {\bf 20}, 223 (1988).}

\bibitem{gouveia_lifts_2011}
	J.\ Gouveia, P.~A.\ Parrilo, and R.\ Thomas,
	\href{http://dx.doi.org/10.1287/moor.1120.0575}
	  {Math.\ Oper.\ Res. {\bf 38}, 248 (2013);}
	\href{http://arxiv.org/abs/1111.3164}{arXiv:1111.3164.}

\bibitem{wolf_nondeterministic_2003}
	R.\ de Wolf, 
	\href{http://dx.doi.org/10.1137/S0097539702407345}
	  {SIAM J.\ Comput. {\bf 32}, 681 (2003).}
	
\bibitem{Cyril13}
	C.~Stark, 
	\href{http://arXiv.org/abs/1210.1105}{arXiv:1210:1105.}
	
\bibitem{cones_and_duality}
	C.~D.\ Aliprantis and R.\ Tourky, 
	{\it Cones and Duality} (American Mathematical Society, Providence, 2007). 
	
\bibitem{FioMasPat13}
	S.\ Fiorini, S.\ Massar, M.~K.\ Patra, and H.\ Raj Tiwary,
	\href{http://arXiv.org/abs/1310.4125}{arXiv:1310.4125.}

\bibitem{Dha63}
	S.~W.\ Dharmadhikari,
	\href{http://dx.doi.org/10.1214/aoms/1177704000}
	  {Ann.\ Math.\ Statist. {\bf 34}, 705 (1963).}

\bibitem{Bar07}
	J.\ Barrett, 
	\href{http://dx.doi.org/10.1103/PhysRevA.75.032304}
	  {Phys.\ Rev.\ A {\bf 75}, 032304 (2007).}
      
\bibitem{BarBarLei06}
	H.\ Barnum, J.\ Barrett, M.\ Leifer, and A.\ Wilce,
	\href{http://arXiv.org/abs/quant-ph/0611295}{quant-ph/0611295.}
      
\bibitem{Hol82}
	A.~S.\ Holevo, 
	{\it Probabilistic and Statistical Aspects of Quantum Theory} 
	(North Holland, Amsterdam, 1982).

\bibitem{GroMueCol10}
	D.\ Gross, M.~P.\ Mueller, R.\ Colbeck, and O.~C.~O.\ Dahlsten,
	\href{http://dx.doi.org/10.1103/PhysRevLett.104.080402}
	  {Phys.\ Rev.\ Lett.\ {\bf 104}, 080402 (2010);}
	\href{http://arxiv.org/abs/0910.1840}{arXiv:0910.1840.}
\end{thebibliography}
\end{document}